%% file: main.tex
\documentclass{llncs}

\usepackage{amsmath, stmaryrd, amssymb}
\usepackage{keybook,keyjavadl,keylistings}
\usepackage{url}
\usepackage{wrapfig}

\newcommand{\concreteDomain}{\ensuremath\mathsf{D}}
\newcommand{\Interpretation}{\ensuremath\mathsf{I}}
\renewcommand{\dep}[1]{\ensuremath{#1^{\mathtt{dep}}}}
\newcommand{\eqterm}{\approx}
\newcommand{\EC}[1]{\ensuremath{[#1]}}
\newcommand{\ECsem}[1]{\ensuremath{[#1]_{\eqterm}}}
\newcommand{\ECS}{\ensuremath\mathsf{EC}}

\newcommand{\policy}{\ensuremath\mathcal{P}}
\newcommand{\declenv}{\ensuremath\mathit{DE}}
\newcommand{\declpairs}{\ensuremath\mathit{DP}}
\newcommand{\depset}{\ensuremath\mathcal{D}}
\newcommand{\lvl}{\ensuremath\mathit{lvl}}
\newcommand{\vars}{\ensuremath\mathit{pvars}}

\newcommand{\deps}{\ensuremath\mathit{deps}}
\newcommand{\wrap}{\ensuremath\mathit{wrap}}

\newcommand{\Sorts}{\ensuremath\mathcal{T}}
\newcommand{\PVars}{\ensuremath\mathcal{PV}}
\newcommand{\Preds}{\ensuremath\mathcal{P}}
\newcommand{\Funcs}{\ensuremath\mathcal{F}}
\newcommand{\LVars}{\ensuremath\mathcal{LV}}

\renewcommand{\val}{\ensuremath\mathit{val}_{M,s,\beta}}

\newcommand{\condterm}[3]{\ensuremath{\mathit{if}({#1})\mathit{then}({#2})\mathit{else}({#3})}}

\newcommand{\draftnote}[1]{}

\begin{document}
\mainmatter
\title{Dependency-Based Information Flow Analysis with Declassification in a Program Logic}
\author {Bart van Delft$^1$
    \and Richard Bubel$^2$ }
\institute {Chalmers University of Technology, Sweden, \email{vandeba@chalmers.se}
        \and Technical University Darmstadt, Germany, \email{bubel@cs.tu-darmstadt.de}}
\maketitle

\input{abstract}
\input{introduction}
\input{syntax_semantics}

\input{security}
\input{calculus}

\input{example}

\input{related}
\input{conclusions}

\bibliographystyle{alpha} 
\bibliography{literature}{}

\newpage
\appendix


\input{app/semantics}
\newpage
\input{app/proofs}

\end{document}

%% file: abstract.tex
\begin{abstract}
We present a deductive approach for the analysis of secure information flows with support for fine-grained policies that include declassifications in the form of delimited information release.
By explicitly tracking the dependencies of program locations as a computation history, we maintain high precision, while avoiding the need for comparing independent program runs. By considering an explicit heap model, we argue that the proposed analysis can straightforwardly be applied on object-oriented programs.
\end{abstract}

%% file: introduction.tex
\section{Introduction}
\label{sec:introduction}

An information-flow analysis tracks where and how information flows within an application to control which inputs may or may not affect what outputs.
The typical information-flow  policy that no information may flow from sensitive locations to less sensitive ones~\cite{Cohen1977} is known as \emph{non-interference}~\cite{goguen1982security}.
In practice this policy turns out to be too strong for many programs: even a simple password checker needs to reveal whether the entered password is equal to the actual password or not.
It is therefore common to enforce a more liberal policy that allows for controlled release of information: a policy with declassifications.

In our previous example, a declassification policy would specify that the outcome of the equality check with the actual password may be leaked, but no other information (in particular, the actual password must not be leaked).
This kind of declassification can be refined with conditions stating under which circumstances certain information may be made public and is called \emph{delimited information release}.\draftnote{we refer to it as?}
In this paper we consider a (very) simplified program that decrypts a ciphertext but only returns the result if it has a valid padding and the integrity check, i.e. the checksum, succeeds.
This program, shown in Fig.~\ref{fig:decrypt}, complies with a policy with two declassifications: we unconditionally declassify whether both checks hold (but not whether the \emph{individual} checks hold), and we declassify the decrypted message on the condition that both checks hold.

Information-flow analysis can be rephrased as determining the dependencies that a program introduces between the input and output locations~\cite{Cohen1977}.
That is, determining whether sensitive information flows to a less sensitive location can be rephrased as determining whether the final value of that less sensitive location depends on any sensitive input.

We present a logic-based approach to information-flow analysis, which improves on earlier work~\cite{BHW09} by one of the co-authors, which focused on automating software verification by integrating abstract interpretation to generate loop invariants. Non-interference checking was used as a show-case for their framework. Their logic allowed to identify the dependencies in commands as \java{x := y; z := x} ({\tt z} depends on {\tt y}, not on {\tt x}) and {\tt x := y; x := 8} ({\tt x} does not depend on anything) which are typically over-approximated by e.g. flow-insensitive type systems for information-flow~\cite{broberg2013paragon,HuntS06,VolpanoS1997}.

However, their approach tracks dependencies on memory locations  directly and performs over-approximations when syntactically inspecting the variables occurring in expressions.  The logic presented in this paper increases precision by maintaining some information on how variable dependencies contributed to the value of an expression and allows us to formulate the desired fine-grained security conditions. In more detail our contributions are:
\begin{figure}[t]
  \centering
\begin{minipage}[b]{0.50\textwidth}
\fontsize{8}{8}\nolstbar\begin{lstjava}
msg = cipher * key;
paddingOk = msg 
if (paddingOk) {
    checkSum = (msg / 256) 
    if (checkSum != -1) {
        result = msg;
    } else {
        result = -1;
    }
} else {
    result = -1;
}
\end{lstjava}
    \caption{Simple 'decryption' algorithm:
Decrypted ciphertext is only revealed if padding \emph{and} checksum of \code{msg} are valid.}
    \label{fig:decrypt}
  \end{minipage}
 \hfill
  \begin{minipage}[b]{0.465\textwidth}\centering
\fontsize{8}{8}\nolstbar\begin{lstjava}
objA.f = 1;
objB.f = 1;
if (h == 0) {
  objC = objA; 
} else {
  objC = objB;
}
objC.f = 17;
\end{lstjava}
    \caption{Information-flow via aliasing: Information flows from~{\tt h} to {\tt objA.f} without an assignment using variable {\tt objA}.}
    \label{fig:alias}
  \end{minipage}
\end{figure}
%
(i) Tracking of dependencies using a term structure instead of sets of variables to provide higher accuracy of the information-flow analysis:
    Approaches like ~\cite{BHW09} add for programs like \java{x=y-y} or \java{if (y) \{x=8\} else \{x=8\}} variable {\tt y} as dependency of {\tt x}, but our analysis does not.
(ii) Support for precise conditional declassification policies:
      For the program in Fig.~\ref{fig:decrypt} we can specify that \java{result} may depend on \java{cipher * key} only if both padding and integrity check succeed.
(iii) Extensible to an object-oriented setting (hints are given in the paper) to analyse complex information flows as shown in Fig.~\ref{fig:alias}.
%
The presented program logic can be integrated into the  loop invariant generation algorithm presented in~\cite{BHW09}.
\paragraph{Paper outline.}
We first introduce  syntax and semantics of both programs and the program logic in Sect.~\ref{sec:syntax:semantics}, as well as a sequent calculus to prove the validity of properties specified in the program logic.
In Sect.~\ref{sec:security} we extend the semantics of the program logic to encode security conditions (non-interference and delimited information release) as program logic formulas.
Sect.~\ref{sec:calculus} presents an extension to the sequent calculus to prove the validity of these formulas.
In Sect.~\ref{sec:example} the program logic and calculus are applied to the programs in Figs.~\ref{fig:decrypt} and~\ref{fig:alias}.
The Sects.~\ref{sec:related} and~\ref{sec:conclusions} conclude the paper by discussing related and future work.

%% file: syntax_semantics.tex
\section{Syntax and semantics}
\label{sec:syntax:semantics}

This section introduces an imperative programming language, a program logic to specify general functional properties and a calculus to verify that a program satisfies its specification (mostly taken from~\cite{BHW09}). The formal treatment of our extensions is based on these definitions. At the end of the section we extend our programming language by adding object-oriented features in a semi-formal manner,  which serves us to demonstrate how to extend our approach to a programming language with a heap.

We define now the syntax of our programming language and logic. The signature of the logic is defined as $
\Sigma =\{\Sorts, \alpha, \Preds, \Funcs, \PVars, \LVars\}
$ with 
$T \in \Sorts$ a sort name; $\alpha$ assigns an arity to
(i)  each \emph{predicate symbol} $p \in \Preds$ as $\alpha(p)=T_1\times\ldots\times T_n$, 
(ii)  each \emph{function symbol} $f \in \Funcs$ as $\alpha(p)=T_1\times\ldots\times T_n\rightarrow T$ (constants are modelled as functions with no arguments),
(iii)  each \emph{program variable} ${\tt x} \in \PVars$ as $\alpha({\tt x}) \in \Sorts$ and
(iv)  each \emph{logical variable} $y \in \LVars$ as $\alpha(y)\in\Sorts$.
  
Logical variables can be bound by quantifiers and are \emph{not} allowed to occur in programs. Logical variables and all function symbols are \emph{rigid} symbols, whose value cannot be changed by programs.
In contrast, program variables cannot be bound by quantifiers, but are allowed to occur in programs.
Program variables are \emph{non-rigid} symbols, i.e., their value is state-dependent and can be changed by program execution.  Further, there are predefined sorts $\top\in\Sorts$ and $\java{int}\in\Sorts$  and predefined predicate symbols $<,\leq,\geq,>$ and function symbols $+,-,\times,/,\%$ with their canonical arithmetic signature.

\begin{definition}[Syntax]
Program {\tt p}, updates $u$, terms $t$ and predicates $\varphi$  are defined as all well-typed words generated by the following grammar:
\begin{align*}
  {\tt p} & ::= {\tt x} = t ~|~ \code{if (}\varphi\code{) \{ p \} else \{ p \}} ~|~ \code{while (}\varphi\code{) \{p\}} ~|~ \code{p; p} \\
  u & ::= \elUp{{\tt x}}{t} ~|~ \parUp{u}{u} \\
  t & ::= {\tt x} ~|~ y ~|~ f(t,\ldots,t) ~|~  t \bullet t~(\bullet\in\{+,-,\times,/,\%\})  ~|~ \condterm{\varphi}{t}{t} ~|~ \applyUp{u}{t} \\
  \varphi & ::= \keytrue ~|~ \keyfalse ~|~ p(t,\ldots,t) ~|~  t \doteq t  ~|~ \keynot\varphi ~|~ \varphi \circ \varphi ~|~ \exists y.\varphi ~|~ \forall y.\varphi  ~|~ \\
               & \phantom{::=\ } \condterm{\varphi}{\varphi}{\varphi} ~|~ \applyUp{u}{\varphi} ~|~ \dlbox{\tt p}{\varphi}   \\
  &\text{where} \quad \circ\in\{\keyand,\keyor,\keyimplies\}, y \in\LVars, {\tt x} \in\PVars, f \in \Funcs, p \in \Preds
\end{align*}
Terms and formulas that appear inside programs may not contain any logical variables, quantifiers, updates, or nested programs. 
\end{definition}

Like programs, updates can be used to describe modifications of the non-rigid symbols (i.e. program variables).
That is, $\dlbox{\tt p}{\varphi}$ denotes that $\varphi$ should be evaluated in a context where the value of the program variables has changed according to {\tt p}.
Similarly, in $\applyUp{u}{\varphi}$ the update $u$ describes how the values of program variables change before evaluating $\varphi$. Intuitively, updates can be seen as a kind of explicit generalized substitutions, which are simplified during and their application is delayed until the program has been symbolically executed.

An elementary update $u$ is a pair $\elUp{{\tt x}}{t}$ where {\tt x} is a program variable and $t$ a term of the same sort as {\tt x}.
Elementary updates $u_1, u_2$ can be combined to parallel updates $\parUp{u_1}{u_2}$.
\begin{example} To explain the intuition behind updates we give a few examples.
\begin{itemize}
\item $\applyUp{\elUp{\java{i}}{\java{a}+1}}{\java{i}>\java{a}}$ expresses that the program variable $\java{i}$ is greater than $\java{a}$ when evaluated in a state where $\java{i}$ has been set to ${\tt a} + 1$.
\item ${\it oldi}\doteq {\it oldj} \limplies \applyUp{\parUp{\elUp{\java{i}}{\java{j}}}{\elUp{\java{j}}{\java{i}}}}{(\java{i}\doteq {\it oldj}\land \java{j}\doteq {\it oldi})}$ expresses that the values of \java{i} and \java{j} are swapped. Note: Parallel updates are `executed' simultaneously without effecting each other. 
\item $\applyUp{\parUp{\elUp{\java{x}}{4}}{\elUp{\java{x}}{8}}}{\java{x} \doteq 8}$ evaluates to true; i.e. in case of competing updates, the right-most update wins.
\end{itemize}
\end{example}

To give semantics to dynamic logic formulas, we require a first-order structure to give meaning to sorts, functions and predicates, a state to give meaning to program variables and a variable assignment to give meaning to logical variables.

A \emph{first-order structure} $M$ is a pair $(\concreteDomain,\Interpretation)$ where $\concreteDomain$ is a domain, i.e. a non-empty set of elements.
The interpretation function $\Interpretation$ assigns
\begin{itemize}
\item each sort $T\in\Sorts$ a non-empty set $\Interpretation(T)\subseteq\concreteDomain$ of elements. 
\item each function symbol $f:T_1\times\ldots\times T_n\rightarrow T\in\Funcs$ a total function $\Interpretation(f):\Interpretation(T_1)\times\ldots\times\Interpretation(T_n)\rightarrow\Interpretation(T)$; and 
\item each predicate symbol $p:T_1\times\ldots\times T_n\in\Preds$ a relation $\Interpretation(p)\subseteq\Interpretation(T_1)\times\ldots\times\Interpretation(T_n)$
\end{itemize}
In particular, $I$ is fixed for sorts $\top$ and $\java{int}$ such that $I(\top)=\concreteDomain$ and $I(\java{int})=\mathbb{Z}$. In addition $\Interpretation$ is also fixed w.r.t\ equality $\Interpretation(\doteq)(v_1,v_2)=\semtrue$ iff.\ $v1=v2$ and functions like $+,-,\times,/,\%$ which are interpreted as the standard arithmetic operations.

Given an interpretation function $\Interpretation$, a \emph{state} $s$ assigns each program variable $\code{x}\in\PVars$ of sort $T$ a value $s(\code{x})\in\Interpretation(T)$; the set of all states is denoted $\States$.
A \emph{variable assignment} $\beta$ assigns each logical variable $y \in \LVars$ of sort $T$ a value $\beta(y)\in\Interpretation(T)$.

\begin{definition}[Semantics]\footnote{$\val$ is defined formally in Appendix~\ref{app:semantics:basic}} 
Given a first-order structure $M$, a state $s$ and a variable assignment $\beta$,
the semantics is defined by  function $\val$ evaluating
\begin{itemize}
  \item terms to a value $\val(t) \in \concreteDomain$,
  \item formulas to a truth value $\val(\varphi) \in \{\semtrue, \semfalse\}$,
  \item updates to a result state $\val(u) \in \States$, and
  \item programs to a set of states $\val({\tt p}) \in 2^{\States}$ with cardinality either 0 or 1.
\end{itemize}
A formula $\varphi$ is called \emph{valid} iff.\ $\val(\varphi) = \semtrue$ for all first-order structures $M$, states $s$ and variable assignments~$\beta$.
\end{definition}

A diverging program evaluates to the empty set. As our programming language is deterministic, a terminating program has exactly one result state.

\begin{example} Example evaluations with $s = \{ \java{x} \mapsto 4 \}$ and $\beta = \{ y \mapsto 3\}$. \\[-1.5em]
\begin{itemize}
\item $\val(y + 3) = 6$
\item $\val(\dlbox{\java{x = 1; x = 2}}{\java{x} \doteq 1}) = \semfalse$
\item $\val(\elUp{\java{x}}{\java{x} + 1}) = \{\java{x} \mapsto 5\}$
\item $\val(\java{x = 1; while (x > 0) \{ x = x + 1; \}}) = \emptyset$
\item $\val(\java{x = 1; while (x > 0) \{ x = x - 1; \}}) = \{ \{ \java{x} \mapsto 0 \} \}$\draftnote{The $\{\}$ notation is used for updates, normal sets, mappings and sets of dependencies...}
\end{itemize}
\end{example}

As mentioned earlier, we add now some object-oriented features to our programming language. This allows us later to demonstrate how to extend our approach to programming language with heaps (and hence, aliasing). The syntax of our programming language is extended by the following three statements: 
\[
{\tt p} ::= ...~|~{\tt x.f} = t~|~{\tt x} = t.{\tt f}~|~{\tt x} = {\tt new} 
\]
which allow to assign a value to the field $f$ of an object $x$, to read the value of an object field and to create an object. Fields are simple names, where different names identify different fields. For the modelling on the logic level we require the existence of four additional sorts: \java{Object} (the sort of all objects), \java{Field} (whose domain elements represent fields),  a \java{Heap} datatype to model the heap and the sort \java{HeapValue} being the smallest supersort of all sorts except \java{Heap}. The heap is modelled using the theory of arrays, i.e., there are two additional function symbols $\mathtt{store}:\java{Heap}\times\java{Object}\times\java{Field}\times\java{HeapValue}\rightarrow\java{Heap}$ and $\mathtt{select}:\java{Heap}\times\java{Object}\times\java{Field}\rightarrow\java{HeapValue}$. Further there is a global program variable $\java{heap}\in\PVars$, which refers to the current heap and is updated or queried each time a program writes to or reads from the heap. 
\begin{example} Given two program variables ${\tt o},{\tt u}\in\PVars$ of sort \java{Object} and two unique constants \java{a}, \java{b} of sort \java{Field} (unique constants refer to different domain elements iff.\ their name is different), then the term
	\begin{itemize}
		\item $\mathtt{store}(\java{heap}, \java{o}, \java{a}, 5)$ represents the updated \java{heap} after assigning \java{o.a} the value $5$.
		\item $\mathtt{select}(\mathtt{store}(\java{heap}, \java{o}, \java{a}, 5), \java{u}, \java{b})$ retrieves the value of \java{u.b} from the heap given as first argument. Since \java{a} and \java{b} refer to different fields the term is equivalent to  $\mathtt{select}(\java{heap}, \java{u}, \java{b})$.
	\end{itemize}	
\end{example}

\subsection{Sequent calculus}
\label{sec:calculus:intro}

To reason about the validity of formulas in our logic, we use a Gentzen-style sequent calculus~\cite{Gentzen35}. In particular, we build upon the calculus used for JavaCardDL \cite{BeckertHS07}. 
The basic data structure for our calculus is a \emph{sequent} $\sequent{}{}$ where $\Gamma,\Delta$ are sets of formulas. A sequent is valid iff.\ the formula $\bigwedge_{\phi\in\Gamma}\phi\rightarrow\bigvee_{\psi\in\delta}\psi$ is valid. A \emph{sequent calculus rule} is an instance of the rule schemata:
\[
\seqRule{name}{\overbrace{\seq{\Gamma_1}{\Delta_1}\quad\ldots\quad\seq{\Gamma_n}{\Delta_n}}^{\mathit{premises}}}{\underbrace{\sequent{}{}}_{\mathit{conclusion}}}
\]
Rules with an empty premise are called \emph{axioms}. A sequent proof is a tree where each node is labeled with a sequent and the root node is labeled with the sequent to be proven. Each node (except the root) is the result of a rule application. Let node $n$ be the $i$-th child of node $p$, then there is a rule $r$ such that the conclusion of $r$ matches the sequent of the parent of $n$ and the sequent of $n$ is equal to the instantiated $i$-th premise of $r$. A branch is \emph{closed} if on one of its nodes (usually the last one) an axiom has been applied. A proof is \emph{closed} if all its branches are closed. The calculus is sound, if it only allows to construct closed proofs when the root sequent is valid. 

The calculus rules for first-order logic are standard and not given here, but we explain the basic idea for the rules dealing with programs (see~\cite{BeckertHS07} for more details). Our calculus follows the symbolic execution paradigm and models a symbolic interpreter for the programming language. The sequent rule for a conditional statement looks like
$$\footnotesize
\seqRule{\scriptsize conditional}
{\sequent{\varphi}{\dlbox{\code{p;}\ldots}{\phi}}\qquad
	\sequent{\code{!}\varphi}{\dlbox{\code{q;}\ldots}{\phi}}
}{
\sequent{}{\dlbox{\code{if (}\varphi\code{) \{p;\} else \{q;\};} \ldots}{\phi}}
}
$$
which splits the proof into two parts, one where the condition is assumed to hold and consequently, the then-branch of the conditional statement is executed followed by the remaining program and the second one dealing with the case that the condition evaluates to false. 

We give now the assignment rules for local variables and heaps 
$$\footnotesize
\begin{array}{c@{\hspace*{0em}}c}
\seqRuleW{\scriptsize assign_{\mathit{local}}}
{
  \sequent{}{\applyUp{u}{\applyUp{\elUp{\java{v}}{t}}{\dlbox{\code{\ldots}}{\phi}}}}
}{
\sequent{}{\applyUp{u}{\dlbox{\code{v =}~t\code{; \ldots}}{\phi}}}
} & 
\seqRuleW{\scriptsize assign_{\mathit{field}}}
{
  \sequent{}{ 
			\applyUp{u}{\applyUp{\elUp{\java{heap}}{\java{store(heap,o,a,}t\code{}}}{\dlbox{\code{\ldots}}{\phi}}}
  }
}
{
	\sequent{}{\applyUp{u}{\dlbox{\code{o.a =}~t\code{; \ldots}}{\phi}}}
}
\end{array}
$$
The  rule $\ruleName{assign_{\mathit{local}}}$ turns an assignment to a local variable directly into an update, while  rule $\ruleName{assign_{field}}$  updates the heap. During symbolic execution updates accumulate in front of the program modality until symbolic execution finishes and are then applied to the postcondition similar to a substitution. This delayed application allows us to perform simplifications on updates before their application and avoids case splits as well as introduction of fresh variables after each assignment. The calculus rules for update simplification and application can be found in~\cite{BeckertHS07}.

%% file: security.tex
\section{Security conditions}
\label{sec:security}

In this section we define information-flow policies  as formulas in our logic. First  a definition based on program variable dependencies for non-interference following the work in~\cite{BHW09} is given, which is then extended to express the more fine-grained conditions of delimited information release.  

\subsection{Non-interference as a program logic formula}

With the semantics in place, we can give a more formal definition of non-interference.
An observer with security level $l$ can observe all information of levels $l'$ `below' $l$, that is information labelled with level $l' \sqsubseteq l$.
Given a policy $\policy = (L, \sqsubseteq)$ (a set of levels $L$ with a partial relation $\sqsubseteq$ between them) we assume a mapping $\mathit{lvl}$ from program variables to their intended security level.
This leads to the following definition of equivalent states for an observer on level $l$:

\begin{definition}[$l$-equivalent states]
  Given $\policy = (L,\sqsubseteq)$, let $l$ be a security level in $L$.
  Two states $s_1, s_2$ are $l$-equivalent, written $s_1 \approx_l s_2$, iff.\  for all program variables ${\tt x}$ with $\mathit{lvl}({\tt x}) \sqsubseteq l$, it holds that $s_1({\tt x}) = s_2({\tt x})$.   
\end{definition}
A program {\tt p} is defined non-interfering if it preserves $l$-equivalence for all levels~$l$:
\begin{definition}[Non-interference]
  Let $\policy = (L,\sqsubseteq)$.
  Given a program {\tt p}, any two states $s_1, s_2$, any first-order structures $M$ and any variable assignments $\beta$,
  such that $\textit{val}_{M,s_1,\beta}({\tt p}) = \{s_1'\}$ 
       and  $\textit{val}_{M,s_2,\beta}({\tt p}) = \{s_2'\}$.
  Then {\tt p} is said to be \emph{non-interfering} with respect to $\policy$
   iff for all $l \in L$, when $s_1 \approx_l s_2$ then $s_1' \approx_l s_2'$.\label{def:noninf}
\end{definition}

Definition~\ref{def:noninf} is  \emph{termination insensitive}, because we assume termination of program {\tt p}.
As a consequence, an observer may learn more information via the (non-)termination of {\tt p} than specified by the policy $\policy$~\cite{Askarov+:ESORICS08}.

It is well-known that the non-interference of a program can be phrased in terms of the dependencies that the program introduces between the initial and final values of program variables~\cite{Cohen1977,HuntS06}.
That is, if the final value of {\tt x} depends on the initial value of {\tt y}, we know that there is interference from $\lvl({\tt y})$ to $\lvl({\tt x})$ and we should therefore check that $\lvl({\tt y}) \sqsubseteq \lvl({\tt x})$.

We present as a small generalisation the dependencies of terms, rather than just program variables since we allow for the declassification on the granularity of terms and not just whole variables in the next section.
Clearly, by choosing the term to be a single program variable the following definition is equivalent to dependencies for variables.

\begin{definition}[Term dependencies]
  Given a program {\tt p} and a term $t$.
  The \emph{term dependencies} of a term $t$ under {\tt p} is defined as the smallest set $\depset(t, {\tt p}) \subseteq \PVars$ of program variables such that 
  for all states $s_1$, $s_2$, any first-order structures $M$ and any variable assignments $\beta$,
  with $\textit{val}_{M,s_1,\beta}({\tt p}) = \{s_1'\}$ 
       and  $\textit{val}_{M,s_2,\beta}({\tt p}) = \{s_2'\}$, we have
    if for all ${\tt y} \in \depset(t, {\tt p})$  $s_1({\tt y}) = s_2({\tt y})$, 
    then  $\mathit{val}_{M, s'_1 , \beta}(t) = \mathit{val}_{M, s'_2 , \beta}(t)$.\label{def:termdeps}
\end{definition}
We can now define non-interference in terms of dependencies as argued:
\begin{definition}[Dependency-Based Non-Interference]
  Given a policy $\policy = (L, \sqsubseteq)$.
  A program {\tt p} is non-interfering with respect to $\policy$ iff for each variable ${\tt x} \in \PVars$,
    for all variables ${\tt y} \in \depset({\tt x}, {\tt p})$,  it holds that $\lvl({\tt y}) \sqsubseteq \lvl({\tt x})$.
\label{def:depnoninf}
\end{definition}

\begin{lemma}
  Definitions \ref{def:noninf} and \ref{def:depnoninf} are equivalent.\label{lem:niasdep} Proof in Appendix~\ref{app:proofs:lem:niasdep}.
\end{lemma}

The security condition can be expressed (per program variable) as a program logic formula (see~\cite{BHW09}). For this purpose, a ghost variable $\dep{\tt x}$ is introduced for  each program variable {\tt x}.
Assume that {\tt x}, {\tt y} and {\tt z} are the only variables in program {\tt p}, $\lvl({\tt y}) \sqsubseteq \lvl({\tt x})$ and $\lvl({\tt z}) \not\sqsubseteq \lvl({\tt x})$.
When we extend the function $\mathit{val}$ appropriately to update the state such that $\dep{\tt x}$ contains (an over-approximation of) the set of dependencies of {\tt x},
we can express the security condition as:
\[
  \dep{\tt x} =  \{{\tt x}\} \mathop{\&}
  \dep{\tt y} =  \{{\tt y}\} \mathop{\&}
  \dep{\tt z} =  \{{\tt z}\} \rightarrow
  [{\tt p}]
  \dep{\tt x} \subseteq \{{\tt x}, {\tt y}\}
\]

For a more formal description of how $\mathit{val}$ is extended to over-approximate the dependencies, we refer to~\cite{BHW09}.

\subsection{Non-interference with declassifications}

In this paper we allow for the specification and verification of more fine-grained security conditions than non-interference on program variable dependencies.
As illustrated by the decryption example in Fig.~\ref{fig:decrypt}, we desire more control on \emph{when} information may interfere, and on \emph{what} information from the initial state is released,
as per the dimensions of declassification~\cite{SabelfeldS2005}.

As a means of specifying such conditions, we introduce a \emph{declassification environment} $\declenv$ to the security policy $\policy = (L, \sqsubseteq, \declenv)$.
A declassification environment maps security levels in $L$ to a set of \emph{declassification pairs} consisting of a formula (\emph{when}) and a term ({\emph{what}).
That is, $(\varphi, t) \in \declenv(l)$ specifies that when $\varphi$ holds on the initial state, the information about that initial state as specified by term $t$ may be declassified to level $l$.

For example, $({\tt password} \keyeq {\tt guess},{\tt (salary1 + salary2) / 2}) \in \declenv({\tt low})$ specifies that the average of two salaries may be released to security level {\tt low} if the right password was entered (regardless of the security level of each salary).\draftnote{Give example for SSL example.}

We can define the equivalence of states with respect to a set of declassification pairs $\declpairs$, as agreeing on the information that is declassified whenever both states allow for the declassification.
A similar notion is used in~\cite{banerjee2008expressive,vanhoef2014stateful} for `flowspecs' resp.\  `stateful declassifications'.

\begin{definition}[$\declpairs$-equivalent states]
  Given a set $\declpairs$ of declassification pairs, a first-order structure $M$ and a variable assignment $\beta$.
  Two states $s_1, s_2$ are $\declpairs$-equivalent, written $s_1 \approx_{\declpairs} s_2$,
  iff for all $(\varphi, t) \in \declpairs$, if $\mathit{val}_{M,s_1,\beta}(\varphi) = \mathit{val}_{M,s_2,\beta}(\varphi) = \semtrue$
  then $\mathit{val}_{M,s_1,\beta}(t) = \mathit{val}_{M,s_2,\beta}(t)$.
\end{definition}

Terms/formulas in declassification pairs may not contain programs, updates or logical variables.
Non-interference for a level $l$ in the presence of a declassification environment is then defined straightforwardly by only considering states that are equivalent w.r.t.\  the declassification pairs that declassify to $l$ (or lower).

\begin{definition}[Non-interference with declassifications]
  Let $\policy = (L,\sqsubseteq,\declenv)$.
  Given a program {\tt p}, any two states $s_1, s_2$, any first-order structures $M$ and any variable assignments $\beta$,
  such that $\textit{val}_{M,s_1,\beta}({\tt p}) = \{s_1'\}$ 
       and  $\textit{val}_{M,s_2,\beta}({\tt p}) = \{s_2'\}$.
  Then {\tt p} is said to be \emph{non-interfering} with respect to $\policy$
   iff for all levels $l \in L$ with $\declpairs = \bigcup_{l' \sqsubseteq l} \declenv(l')$,
     when $s_1 \approx_l s_2$ and $s_1 \approx_{\declpairs} s_2$
     then $s_1' \approx_l s_2'$.
\label{def:noninf:decl}
\end{definition}

To specify non-interference with declassifications as a dependency problem, we revisit our definition of term dependencies to be parametrised on a set of declassification pairs $\declpairs$.
Definition~\ref{def:termdeps:decl} varies from Definition~\ref{def:termdeps} only in the additional requirement that $s_1$ and $s_2$ are $\declpairs$-equivalent.

\begin{definition}[Term dependencies with declassifications]
  Given a program {\tt p}, a term $t$ and a set of declassification pairs $\declpairs$.
  The \emph{term dependencies with declassifications} of $t$ under {\tt p} is  defined as the smallest set $\depset(t, {\tt p}, \declpairs) \subseteq \PVars$ of program variables such that 
  the following holds for all states $s_1$, $s_2$, for any first-order structures $M$ and any variable assignments $\beta$ with $\textit{val}_{M,s_1,\beta}({\tt p}) = \{s_1'\}$ and  $\textit{val}_{M,s_2,\beta}({\tt p}) = \{s_2'\}$.
  If $s_1 \approx_{\declpairs} s_2$
  and for all ${\tt y} \in \depset(t, {\tt p}, \declpairs)$  $s_1({\tt y}) = s_2({\tt y})$, 
  then  $\mathit{val}_{M, s'_1 , \beta}(t) = \mathit{val}_{M, s'_2 , \beta}(t)$.
\label{def:termdeps:decl}
\end{definition}

We can now define non-interference in the presence of a declassification environment, by determining the dependencies of a variable {\tt x} in the presence of all declassifications to a level below $\lvl({\tt x})$.

\begin{definition}[Dependency-Based Non-Interference with declassifications]
  Given a policy $\policy = (L, \sqsubseteq, \declenv)$.
  A program {\tt p} is non-interfering with respect to $\policy$ iff for each variable ${\tt x} \in \PVars$,
  with $\declpairs = \bigcup_{l' \sqsubseteq \lvl({\tt x})} \declenv(l')$,
  for all variables ${\tt y} \in \depset({\tt x}, {\tt p}, \declpairs)$,  it holds that $\lvl({\tt y}) \sqsubseteq \lvl({\tt x})$.
\label{def:depnoninf:decl}
\end{definition}

\begin{lemma}
  Definitions \ref{def:noninf:decl} and \ref{def:depnoninf:decl} are equivalent. Proof in App.~\ref{app:proofs:lem:declniasdep}.
  \label{lem:declniasdep}
\end{lemma}

\subsection{Non-interference with declassifications as a formula}

Approximating the dependencies of a program variable {\tt x}  as a set of other variables results in a significant loss of precision, as it eliminates all knowledge  about \emph{how} the dependencies contributed to the value of \java{x}. This makes it impossible to determine whether some dependencies can be removed due to declassification.
Intuitively, we need to know what information about the initial state is recorded in the value of {\tt x} in the final state.
The approach taken here, is to record how the final value of {\tt x} can be computed from the initial state.
Then, we can inspect this computation and determine whether it contains any declassified information.

Naively, we could introduce the dependency variable \dep{\tt x} as done  in~\cite{BHW09}, except now tracking a \emph{term} that represents this specialised computation instead of the explicit set of dependencies.
That is, given a program {\tt p} with $\val({\tt p}) = \{s'\}$, we would have that $\val(s'(\dep{{\tt x}})) = \mathit{val}_{M,s',\beta}({\tt x})$.
Multiple terms can represent the same information, hence, simply encoding this specialised computation as a single term $t$ is problematic.

For example, the term \java{x - x} represents the same computation as \java{0}, but the variables occurring in each term (i.e. the dependencies) are different.
At the same time, \java{x+y} specifies the same information as the term \java{y+x}, but the syntactic term is different.
Therefore representing the computation as one specific term might restrict us to specify the declassification terms in the same syntactic form.


These issues can be addressed by representing the computation not with a specific term, but with an \emph{equivalence class} of terms.
Two terms are said to be equivalent if they both evaluate to the same value in any context.\draftnote{actually, we should exclude $M$ from this for all quantification right? We assume a fixed interpretation for constants and function/predicate symbols.}

\begin{definition}[Equivalent terms]
  Two terms $t$, $t'$ are equivalent, written $t \eqterm t'$, iff.\ 
  for all first-order structures $M$, states $s$ and variable assignments $\beta$,
  $\val(t) = \val(t')$.
\end{definition}

As is common, we denote with $\ECsem{t}$ the equivalence class of term $t$, i.e. the set of all terms equivalent to $t$.\draftnote{This overloads $[{\tt p}]$}
Analogously, we also introduce equivalence classes of formulas, $\ECsem{\varphi}$, but focus in the remainder of this section on terms only.

We extend our syntax with a sort for equivalence classes $\ECS \in \Sorts$.
As terms of this sort we introduce a \emph{Herbrand symbol} for every concrete function, predicate and program variable symbol in the signature.
We denote these symbols using $\EC{\cdot}$ and preserve the arity but substitute all sorts for the universal sort $\ECS$.
For example,
 $\EC{{\tt x}} : \ECS$,
 $\EC{17} : \ECS$,
 $\EC{+} : \ECS \rightarrow \ECS \rightarrow \ECS$ and
 $\EC{\geq} : \ECS \rightarrow \ECS \rightarrow \ECS$.
We naturally extend this syntax for function application, e.g. we write $\EC{+} (\EC{{\tt x}}, \EC{17})$ as $\EC{{\tt x} + 17}$.
The Herbrand terms are then evaluated to the equivalence class that they represents, i.e. $\val(\EC{t}) = \ECsem{t}$.\draftnote{We only define the semantics for Herbrand that are also real terms, e.g. {\tt true + false} is a possible Herbrand term but can't have a semantics like that.. Should we be more precise here?}
This implies that, for example, $\val(\EC{{\tt x} - {\tt x}}) = \val(\EC{0})$.

We update the evaluation function for programs to ensure that $\dep{\tt x}$ maintains the equivalence class of computations representing the information that variable {\tt x} has on the initial state.
We give special treatment to the control statements since we need to capture the implicit dependencies.
For a conditional statement branching on $\varphi$, we evaluate each branch which gives us, per variable {\tt x}, the equivalence classes $e_1$ and $e_2$ for each computation.
We select from each class a representative term $t_1$ and $t_2$ (for our purpose it does not matter which exact term is selected) and update $\dep{\tt x}$ to $\EC{\condterm{\varphi}{t_1}{t_2}}$.
For while loops we introduce dedicated symbols to summarise the computation, over-approximating dependencies as all variables contributing to the loop and its condition.

The updated semantics (given in App.~\ref{app:semantics:extended}) allows us to approximate the dependencies of a program variable {\tt x} by picking \emph{any} term $t \in \dep{\tt x}$ and taking $\vars(t)$, the set of program variables occurring in~$t$.

\begin{lemma}[Correctness of \dep{\tt x}]
\label{lem:xdepcorrect}
For all programs {\tt p} and program variables {\tt x}, 
for any state $s$, first-order structure $M$, and variable assignment $\beta$
with $\val({\tt p}) = \{s'\}$ and $s(\dep{\tt y}) = \ECsem{\tt y}$ for all ${\tt y} \in \PVars$,
it holds that $\depset({\tt x}, {\tt p}) \subseteq \vars(t)$ for each $t \in s'(\dep{\tt x})$.
(Proof: see Appendix~\ref{app:proofs:lem:xdepcorrect})
\end{lemma}
In the absence of declassifications, we can thus phrase the security condition for program variable {\tt x}
with $\lvl({\tt y}) \not\sqsubseteq \lvl({\tt x})$ in program {\tt p} as follows:\draftnote{Not entirely true, since this would allow one to change the term obtained from $\choice$ but before applying {\it pvars}. However, since this is not the final form of the formula I think we don't have to mention that here.}
\[
  \dep{\tt x} \keyeq \EC{\tt x} \keyand
  \dep{\tt y} \keyeq \EC{\tt y} \keyimplies
  \dlboxf{{\tt p}}{
  \vars(\choice(\dep{\tt x})) \subseteq \{{\tt x}\}}
\]
Here $\choice$ selects an arbitrary term from the  equivalence class, i.e.. $\choice(\ECsem{t}) \in \ECsem{t}$, which, by Lemma~\ref{lem:xdepcorrect}, gives a correct over-approximation of the term dependencies.
Our calculus implements a natural implementation of $\choice$, namely to take exactly the term that currently syntactically represents the equivalence class.

A declassification environment conditionally allows us to ignore some of the program variables in the term dependency.
To reflect this, we replace the function $\vars \cdot \choice$ with the undefined function $\deps$.
Let $(\varphi, t) \in \declenv(l)$ with $L$ denoting the set of program variables of level $l$ or lower.
We then define the dependency sets of the equivalence classes by adding the assumption:
\[
  \forall \ECS\ e . \
    \deps(e) \keyeq
      \mathit{if}(\varphi \keyand e \keyeq \EC{t})\mathit{then}( L )\mathit{else}(\mathit{h}(e) )
\]
which states that when the declassification condition $\phi$ holds in the initial state and the  equivalence class $e$ contains the term $t$ (which specifies the declassified information) then $e$'s dependency set is simply the set $L$. Otherwise, the dependencies are over-approximated by the function $h$, which applies $\deps$ on the sub-computations forming $e$.
We define $h$ inductively over the term structure used to represent equivalence classes:
$$ 
  h(\EC{\tt x}) = \{{\tt x}\} \quad
  h(\EC{f(t_1,\ldots,t_n)}) = \deps(t_1) \cup\dots\cup\deps(t_n) \quad
  h(\EC{\keynot\varphi}) = \deps(\varphi)
$$
and $h$  applies $\deps$ similar homomorphic on all other terms and formulas.

Multiple declassifications can be  defined by nesting single declassifications.
For instance, for the decryption program {\tt p} from Fig.~\ref{fig:decrypt}, we can consider the policy $\declenv({\tt Low})$ with two declassification pairs $\{ (\keytrue, \varphi), (\varphi, \java{cipher*key}) \}$.
Here, $\varphi = \varphi_1 \keyand \varphi_2$ with $\varphi_1 = ({\tt cipher} * {\tt key}) \% 256 \keyeq 0$ and $\varphi_2 = \keynot((({\tt cipher}*{\tt key})/256)\%256 \keyeq -1)$, specify together when the decrypted message can be released. Namely, it is always released whether the validity check of the decrypted message was successful or not, but the message itself is only released if it has a valid structure. The adherence of \java{p} to this policy can be encoded as a formula in our logic as follows:
\begin{align}
  & \dep{{\tt cipher}} \keyeq \EC{{\tt cipher}} \keyand
    \dep{{\tt key}}    \keyeq \EC{{\tt key}}    \keyand
    \forall \ECS\ e . \
    \deps(e) \keyeq \nonumber \\
  & ~\mathit{if}(\keytrue \keyand e \keyeq \EC{\varphi})\mathit{then}( \{{\tt res}\} )\mathit{else}( \label{eq:decr:formula} \\ 
  & \quad \mathit{if}(\varphi \keyand e \keyeq \EC{{\tt cipher}*{\tt key}})\mathit{then}( \{{\tt res}\} )\mathit{else}(\mathit{h}(e) )) \nonumber
   \rightarrow [{\tt p}]
      \deps(\dep{\tt res}) \subseteq \{{\tt res}\} \nonumber
\end{align}
The following lemma ensures soundness of the above non-interference encoding:

\begin{lemma}[Soundness]
\label{lem:soundnessNI}
Given a policy $\policy = (L, \sqsubseteq, \declenv)$.
Let ${\it INIT}$ be the conjunction of formulas $\dep{\tt y} =  \EC{\tt y}$ for every ${\tt y} \in \PVars$,
${\it DEPS}$ the formula defining $\deps$ as described above, and
${\it LVL}({\tt x})$ the set of program variables with a level $l' \sqsubseteq \lvl({\tt x})$.
Then for every program {\tt p}, the validity of
\[
  {\it INIT} \keyand
  {\it DEPS} \keyimplies
  [{\tt p}] \deps(\dep{\tt x})) \subseteq {\it LVL}({\tt x})
\]
for each ${\tt x} \in \PVars$ implies non-interference of {\tt p} as per Definition~\ref{def:depnoninf:decl}.
(Proof: By Lemma~\ref{lem:xdepcorrect} and correctness of the {\it DEPS} encoding, see Appendix~\ref{app:proofs:soundness})
\end{lemma}

%% file: calculus.tex
\section{A calculus for dependencies}
\label{sec:calculus}

We present here the sequent calculus for our program logic as defined in Sect.~\ref{sec:security}. The calculus follows the symbolic execution paradigm, but with some changes compared to the calculus defined in Sect.~\ref{sec:calculus:intro}.  We start with the calculus rule for dealing with  assignments:
$$\footnotesize
  \seqRule{{assignment}^{\tt dep}}
	{	   \sequent{}
	   {\applyUp{u}{} \upl {\tt x} \upd t \,\|\,
	        \dep{{\tt x}} \upd  \wrap(t) \upr \dlboxf{\dots} \varphi}
	}
	{\sequent
	  {}
	   { \applyUp{u}{} \dlboxf{x = \ensuremath{t}; \dots}\varphi}
	}
$$
As we can see the rule is almost unchanged to the version given in Sect.~\ref{sec:calculus:intro} except that in addition to updating the value of \java{x} also the value of the ghost variable \dep{\java{x}} is updated to capture the dependencies of \java{x} faithfully. The value of  \dep{\java{x}}  is updated to the Herbrand term computed from $t$ by function $\wrap$. Note, $\wrap$ is not a function of the logic, and hence, not part of the syntax -- instead it is a meta function which given a term returns a term. For instance, let $t$ above be instantiated with the term $\java{x}+\java{y}$ ($\java{x}$ and $\java{y}$ being program variables) then $\wrap(\java{x}+\java{y})=\EC{+}(\dep{\java{x}}, \dep{\java{y}})$. 
The full definition of \emph{wrap} can be found in Figure~\ref{fig:wrap} in Appendix~\ref{app:semantics:extended}.

The conditional rule is changed more fundamentally for two reasons (i) we need to track implicit dependencies introduced by the condition and (ii) to capture the dependencies as precise as possible:
$$\footnotesize
  \seqRule{{ifElse}^{\tt dep}}
	{	   
	  \Gamma \turnstyle
	   \applyUp{u}{}  \upl \bar{\tt x}^{pre} \upd \bar{\tt x} \upr  \dlboxf{p1}{} 
	   \upl \bar{\tt x}^t \upd \bar{\tt x} \upr \upl r \upr \dlboxf{p2}{} 
	   \upl \bar{\tt x}^e \upd \bar{\tt x} \upr \upl r \upr 
	        \upl v \upr          \dlboxf{\dots} \varphi. \Delta
	}
	{\sequent
	  {}
	   {\applyUp{u}{} \dlboxf{if(\ensuremath{\varphi})\{p1\}else\{p2\}; \dots}\varphi}
	}
$$
where
\begin{itemize}
  \item $\bar{\tt x} = ({\tt x}_1, {\tt x}_1^{\tt dep}, \dots, {\tt x}_n, {\tt x}_n^{\tt dep})$ 
        are all variables changed in {\tt p1} or {\tt p2} with their corresponding dependency variables,
  
  \item $\bar{\tt x}^{pre} = ({\tt x}^{pre}_1, {\tt x}^{predep}_1,\dots, {\tt x}^{pre}_n, {\tt x}^{predep}_n)$,
        $\bar{\tt x}^t = ({\tt x}^t_1, \dots, {\tt x}^t_n)$ and 
        $\bar{\tt x}^e = ({\tt x}^e_1, \dots, {\tt x}^e_n)$ are lists of fresh variables,
  
 \item $r$ is the resetting update $\bar{{\tt x}} \upd \bar{\tt x}^{pre}$, and
  
 \item $v$ is the collection of parallel updates for all $ {\tt x}_i, {\tt x}_i^{\tt dep} \in \bar{\tt x}$:\hfill\\[0.25em]
$
\quad         {\tt x}_i       \upd{}  \condterm{\varphi}{{\tt x}^t_i}{{\tt x}^e_i}~~\|~~
         \dep{{\tt x}_i} \upd{}  \EC{\condterm{\wrap(\varphi)}{{\tt x}^{t, {\tt dep}}_i}{{\tt x}^{e, {\tt dep}}_i}}
$
\end{itemize}
In contrast to the standard rule, the above rule does not split the proof into two branches. Instead it computes first the effect of the then-branch  \java{p1} and stores the values of the modified variables by assigning them to the fresh variables $\bar{\tt x}^t$. Afterwards the values of the modified variables are set back to their values before execution of \java{p1}. Next the else-branch \java{p2} is executed, the values of the modified variables is saved in $\bar{\tt x}^e$ and again reset. Finally, the values stored in $\bar{\tt x}^t$ and $\bar{\tt x}^e$ are combined (in update $v$) using conditional terms such that they describe the final values of the variables $\bar{\tt x}$ and $\overline{\dep{x}}$. 
\begin{example} Assume the  sequent
	$\footnotesize
		\seq{}{\dlboxf{if (b) \{x=1\} else \{x=2\}; \dots}\varphi}		
	$.
Applying rule \ruleName{{ifElse}^{\tt dep}} we get the sequent (\java{x}, \java{b} are program variables)
\[\footnotesize
\begin{array}{l}
		\begin{array}{l}\Longrightarrow
				\applyUp{\java{x}^{\tt pre}\upd \java{x}\|\java{x}^{\tt predep}\upd \dep{\java{x}}}{\dlboxf{x=1}{\applyUp{\java{x}^{\tt t}\upd \java{x}\|\java{x}^{\tt tdep}\upd \dep{\java{x}}}{\applyUp{\java{x}\upd \java{x}^{\tt pre}\|\dep{\java{x}} \upd \java{x}^{\tt predep}}{}}}}
				\\\phantom{\Longrightarrow}
			       \quad\dlboxf{x=2}{
			       	          \applyUp{\java{x}^{\tt e}\upd\java{x}\|\java{x}^{\tt edep}\upd \dep{\java{x}}}{\applyUp{\java{x}\upd \java{x}^{\tt pre}\|\dep{\java{x}} \upd \java{x}^{\tt predep}}{}}}\\
			       \phantom{\Longrightarrow}\qquad{\applyUp{\java{x}\upd\condterm{\java{b}}{\java{x}^{\tt t}}{\java{x}^{\tt e}}\| \dep{\java{x}}\upd\EC{\condterm{}{}{}}(\dep{\java{b}}, \java{x}^{\tt tdep},\java{x}^{\tt edep})}{\dlbox{\ldots}{\varphi}}}
		\end{array}
\end{array}
\]
After applying $\dep{\ruleName{assignment}}$ twice and simplification the above sequent becomes
\[\footnotesize
\begin{array}{l}
\Longrightarrow
\quad\begin{array}{l}
\applyUp{\java{x}\upd\condterm{\java{b}}{\java{1}}{\java{2}}\| \dep{\java{x}}\upd\EC{\condterm{}{}{}}(\dep{\java{b}}, \EC{1},\EC{2})}{\dlbox{\ldots}{\varphi}}
\end{array}
\end{array}
\]
\end{example}	

The following rule demonstrates how to reason about dependencies of equivalence classes whose representative is a conditional term:
\[\footnotesize
  \seqRule{{ifElseSplit}^{\tt dep}}
  {
    \sequent{}{\deps(\EC{\varphi}) \cup \deps(\condterm{\varphi}{\EC{t_1}}{\EC{t_2}}) \subseteq L}
  }
  {
    \sequent{}{\deps(\EC{\condterm{\varphi}{t_1}{t_2}}) \subseteq L}
  }
\]
The above rule is sound as the conditional term in the premise causes the proof to split into two branches in which the dependencies of the then-branch resp.\ else-branch are included.

In addition to the above rules we need some rewrite rules for equivalence classes, for instance, $\EC{x + 0} \leadsto \EC{x}$ or $\EC{\condterm{\varphi}{t}{t}} \leadsto \EC{t}$ which allows us to prove that two equivalence classes are equal. 

As it turns out the presented approach does not require a special loop invariant rule, but can use the classic loop invariant rule shown below
\[\footnotesize
  \seqRule{loopInvariant}
  {
    \sequent{}{\upl u \upr \mathit{Inv}} \qquad
    \mathit{Inv}, \varphi \turnstyle \dlboxf{p}{\mathit{Inv}}\qquad
    \mathit{Inv}, \keynot \varphi \turnstyle \dlboxf{...}{\phi}
  }
  {
    \sequent{}{\upl u \upr \dlboxf{\java{while}(\ensuremath{\varphi})\{p\};...}{\phi}} \\
  }
\]
The loop invariant rule splits the proof into three branches: (i) in the first branch one has to show that the invariant $\mathit{Inv}$ holds just before entering the loop; (ii) the second branch requires to prove that the loop body preserves the invariant, i.e., assuming the loop condition $\varphi$ and invariant hold before execution of the loop body then both hold again after the execution; and finally in (iii) we have to prove the original post condition $\phi$ when exiting the loop using only the assumptions that the invariant holds and the loop guard was evaluated to false. The only thing which changes in our logic is that the loop invariant formula $\mathit{Inv}$ must also state invariants about the dependency variables $\dep{x}$. 

%% file: example.tex
\section{Example}
\label{sec:example}


\subsection{Decryption}

Let ${\it INIT} \keyand
  {\it DEPS} \keyimplies
  \dlboxf{{\tt p}}{\deps(\dep{\tt res}) \subseteq \{{\tt res}\}}$
be the summarised formula of equation~\eqref{eq:decr:formula} from Sect.~\ref{sec:security}
specifying the non-interference for the decryption program from Fig.~\ref{fig:decrypt}. 
Abbreviating the antecedent with $\Gamma$ and the formula after the program with $\psi$,
we get as initial sequent:
\begin{center}
\(  \Gamma \seq{}{}
  \dlboxf{{\small \code{msg = cipher * key; paddingOk = msg \% 256 == 0; ...}}}{\psi}
\)
\end{center}
Applying \dep{\ruleName{assignment}} twice and some update rewrite rules, we get:
\[
  \Gamma \seq{}{}
  \applyUp{u}{\dlboxf{{\small \code{if(paddingOk)\{..\}else\{...\}}}}{\psi}}
\]
Where, $u$ contains the updates $\elUp {\dep{{\tt msg}}} {\EC{{\tt cipher} * {\tt key}}}$
and $\elUp {\dep{{\tt paddingOk}}} {\EC{({\tt cipher} * {\tt key}) \% 256}}$.
After two applications of \dep{\ruleName{ifElse}} and three \dep{\ruleName{assignment}};
we get after update simplification and applying the update on $\psi$:
\[
  \Gamma \seq{}{}
  \deps(\EC{
    {\it if}(\varphi_1)
    {\it then}(
      {\it if}(\varphi_2)
      {\it then}({\tt cipher} * {\tt key})
      {\it else}(-1)
    )
    {\it else}(-1)
  }) \subseteq \{{\tt res}\}
\]
where $\varphi_1 = ({\tt cipher} * {\tt key}) \% 256$ and $\varphi_2 = (({\tt cipher} * {\tt key})/256) \% 256$.
The rewrite rules for equivalence classes allow us to replace the syntactic term with\draftnote{Check if right rules are included/used}:
\[
  \Gamma \seq{}{}
  \deps(\EC{
    {\it if}(\varphi_1 \keyand \varphi_2)
    {\it then}({\tt cipher} * {\tt key})
    {\it else}(
      {\it if}(\varphi_1)
      {\it then}(-1)
      {\it else}(-1)
    )
  }) \subseteq \{{\tt res}\}
\]
Which in turn can be rewritten to (also replacing $\varphi_1 \keyand \varphi_2$ with the equal $\varphi$):
\[
  \Gamma \seq{}{}
  \deps(\EC{
    {\it if}(\varphi)
    {\it then}({\tt cipher} * {\tt key})
    {\it else}(-1)
  }) \subseteq \{{\tt res}\}
\]
We can then apply the rule $\ruleName{ifElseSplit}^{\tt dep}$:
\[
  \Gamma \seq{}{} {\deps(\EC{\varphi}) \cup \deps(\condterm{\varphi}{\EC{{\tt cipher} * {\tt key}}}{\EC{-1}}) \subseteq \{{\tt res}\}}
\]
First, we observe that $\Gamma$ allows us to derive that $\deps(\EC{\varphi}) \keyeq \{{\tt res}\}$.
Then, we can split on the formula $\varphi$:
\begin{align*}
   \Gamma, \varphi &\seq{}{}
   \{{\tt res}\}
  \cup \deps(\EC{{\tt cipher} * {\tt key}})
    \subseteq \{{\tt res}\} \\
  \Gamma, \keynot \varphi &\seq{}{}
   \{{\tt res}\}
  \cup \deps(\EC{-1})
    \subseteq \{{\tt res}\}
\end{align*}
For the first sequent, we can now derive that $\deps(\EC{{\tt cipher} * {\tt key}}) \keyeq \{{\tt res}\}$, while
in case of the second sequent, we  get  $\deps(\EC{-1}) = h(\EC{-1}) = \emptyset$ 
$$
   \Gamma, \varphi \seq{}{}
  \{{\tt res}\}  \subseteq \{{\tt res}\} \quad \text{and} \quad
  \Gamma, \keynot \varphi \seq{}{}
  \{{\tt res}\}  \subseteq \{{\tt res}\}
$$
Both are trivially valid and thus proving that the program from Fig.~\ref{fig:decrypt} satisfies non-interference under the specified declassification policy.

\subsection{Loop invariants}

\begin{wrapfigure}{r}{0.4\textwidth}
	\label{fig:code:mult}
  \begin{center}
	 \vspace{-3.5em}
	  \nolstbar
    \begin{lstjava}
  m0 = m;
  r = 0;
  while (m > 0) {
	  r = r + n;
	  m = m - 1;
  }
\end{lstjava}
	\vspace{-2em}
  \end{center}
  \caption{Multiplication.}
	\vspace{-1.5em}
\end{wrapfigure}

As is usual in program verification, loops can be dealt with using invariants.
Typically, this requires the interaction with a human user to specify an appropriate invariant.
Consider the program in Fig.~\ref{fig:code:mult} which uses a while loop to multiply the (positive) numbers {\tt m} and {\tt n}, storing the result in {\tt r}.
The variable {\tt m0} is an auxiliary variable to record the pre-state of variable {\tt m}, allowing a human user to specify the invariant $\dep{{\tt r}} \keyeq \EC{({\tt m0} - {\tt m}) * {\tt n}}$.

A strong advantage of the approach presented here is that invariants can be generated automatically using the algorithm presented in~\cite{BHW09}.
For the above program the generated loop invariant would be $\deps(\dep{{\tt r}}) \subseteq \{{\tt m}, {\tt r}, {\tt n} \}$.
Although these invariants drop the computational history of {\tt r}, they appear sufficient for the verification of our security conditions.


\subsection{Heap}


As a final remark, we discuss how an explicitly represented heap can be integrated in our calculus.
By treating the heap like any other program variable with dedicated functions {\tt store} and {\tt select}, our extended calculus can be  applied as-is to  programs with heaps.
Let {\tt p} be the program from Fig.~\ref{fig:alias}. The sequent (without declassifications)
\[
  \Gamma \seq{}{} \dlboxf{p}{\deps(\EC{\mathtt{select}(\dep{\java{heap}},\dep{{\tt objA}}, {\tt f})})} \subseteq \{\}
\]
would be valid, were it not for the possibility of {\tt objA} being an alias of {\tt objC}.
Applying the rules for assignments, conditionals and simplifications, we obtain:
\begin{align*}
  \Gamma \seq{}{}&
    \cdots
    \upl
      \dep{{\tt objC}} \upd \EC{\condterm{{\tt h} \keyeq 0}{\dep{{\tt objA}}}{\dep{{\tt objB}}}}
    \upr\\
  & \deps(\EC{\mathtt{select}(\mathtt{store}(\dep{{\tt heap}}, \dep{{\tt objC}}, {\tt f}, 17),\dep{{\tt objA}}, {\tt f})}) \subseteq \{\}
\end{align*}
Thanks to representing the heap explicitly, the fact that the value of {\tt objC} depends on {\tt h} is also present when reading an element from the heap.
The equivalence class describing the dependencies of \java{objA.f} thus correctly contains the dependency on {\tt h} and the sequent cannot be proven valid.

%% file: related.tex
\section{Related work}
\label{sec:related}

The seminal work~\cite{SabelfeldM04} proposes a type-based approach to delimited information release that requires to mark program points at which declassification occurs explicitly.
Further  type-based systems  have been developed by \cite{BarthePR2007} and \cite{BanerjeeN2005}, but these omit certain features like support for declassification.
	
The most popular type-based system that allows for declassification is the Jif~\cite{Myers1999} approach and tool.
Its rich set of features allows it to be used for realistic programs~\cite{clarkson2007civitas,hicks2006trusted}.
The programming language Paragon\cite{broberg2013paragon} presents an alternative type system that allows for the enforcement of dynamic policies.
That is, rather than explicitly declassifying information, the ordering between security labels can be changed during program execution.

There is also a number of logic-based approaches to information-flow security that make use self-composition~\cite{DarvasHS2003,DarvasHS2005,BartheCK2011}.
In contrast our approach is based on the explicit tracking of dependencies and not on the comparison of two program runs.
This avoids the need for optimizations to prevent the actual repeated execution of two runs like in product programs~\cite{BartheCK2011}.
In addition to self-composition, \cite{DarvasHS2003,DarvasHS2005} also present an alternative  semantic encoding of non-interference, but
which lacks abstraction requiring a higher degree of interaction during the proving process.\draftnote{Is this a comment on all the self-composition work? In that case should be clear.}

\cite{banerjee2008expressive} introduce flowspecs: small declassification specifications on a small group of statements that should be checked using a program verification, interleaved with the type-checking of non-interference.
Since specifications can only be defined on the program's pre-state, this modular approach requires that no assignment is made to variables which are to be declassified at a later point.

In~\cite{HahnlePRW2008} the authors also integrate a type-based system for information flow into a program logic. Their  logic guarantees to find a proof automatically for all programs that are deemed secure by the type-based system. However, their logic has to be designed freshly for any other kind of type-based systems.

The approach presented in~\cite{BHW09} is based on abstract interpretation and explicit dependency tracking. We extended the precision of their approach by tracking not just the sets of variables, but also the computational structure of the dependencies. This allows us to reason about delimited information release properties and not only non-interference properties.


%% file: conclusions.tex
\section{Conclusions and future work}
\label{sec:conclusions}

In this paper we presented a program logic and calculus suitable to express and prove that a program has secure information-flow with respect to non-interference and conditional delimited information release. The program logic was based on explicit dependency tracking instead of self-composition. As future work we plan to integrate our approach in the abstraction framework introduced in~\cite{BHW09} and to perform experiments to evaluate the degree of automation.

%% file: app/semantics.tex
\section{Semantics}
\subsection{Basic Semantics}
\label{app:semantics:basic}

The semantics of terms is defined in Figure~\ref{fig:evaluation-of-terms},
for formulas in Figure~\ref{fig:evaluation-of-formulas},
for programs in Figure~\ref{fig:evaluation-of-programs} and
for updates in Figure~\ref{fig:evaluation-of-updates}.

In Figure~\ref{fig:evaluation-of-formulas}, we denote with $\beta^v_y$ the variable assignment $\beta$ that maps logical variable $y$ to the value $v$ instead of the original value.
In Figure~\ref{fig:evaluation-of-programs}, we use a fixed-point semantics for while loops.
In Figure~\ref{fig:evaluation-of-updates}, to account for the `right-most-wins' semantics, we define the semantics for elementary update and parallel updates simultaneously.

\begin{figure}
\[\begin{array}{rl}
  \val({\tt x}) &= s({\tt x}) \\
  \val(y) &= \beta(y) \\
  \val(f(t_1,\dots,t_n) &= \Interpretation(f)(\val(t_1), \dots, \val(t_n)) \\
  \val({\it if}(\varphi){\it then}(t_1){\it else}(t_2)) &= \begin{cases}
      \val(t_1) & \text{if } \val(\varphi) = \semtrue \\
      \val(t_2) & \text{otherwise}
    \end{cases} \\
  \val(\applyUp{u}{t}) &= {\it val}_{M, s', \beta}(t) \text{ where } s' = \val(u)
\end{array}\]
\caption{Definition of the evaluation function $\val$ for terms.}
\label{fig:evaluation-of-terms}
\end{figure}

\begin{figure}
\[\begin{array}{rl}
  \val(\keytrue) &= \semtrue \\
  \val(\keyfalse) &= \semfalse \\
  \val(p(t_1,\ldots,t_n)) &= \semtrue \text{ iff } (\val(t_1),\ldots,\val(t_n)) \in \Interpretation(p)\\
  \val(\keynot\varphi) &= \semtrue \text{ iff } \val(\varphi) = \semfalse \\
  \val(\varphi_1 \keyand \varphi_2) &= \semtrue \text{ iff } \semfalse \not\in \{\val(\varphi_1), \val(\varphi_2)\}\\
  \val(\varphi_1 \keyor \varphi_2) &= \semtrue \text{ iff } \semtrue \in \{\val(\varphi_1), \val(\varphi_2)\} \\
  \val(\varphi_1 \keyimplies \varphi_2) &= \val(\keynot \varphi_1 \keyor \varphi_2)\\
  \val(\exists y.\varphi) &= \semtrue \text{ iff } \semtrue \in \{{\it val}_{M,s,\beta^v_y}(\varphi) \mid v \in \concreteDomain \}\\
  \val(\forall y.\varphi) &= \semtrue \text{ iff } \semfalse \not\in \{{\it val}_{M,s,\beta^v_y}(\varphi) \mid v \in \concreteDomain \}\\
  \val(\condterm{\varphi_1}{\varphi_2}{\varphi_3}) &= \begin{cases}
      \val(\varphi_2) & \text{if } \val(\varphi_1) = \semtrue \\
      \val(\varphi_3) & \text{otherwise}
    \end{cases} \\
  \val(\applyUp{u}{\varphi}) &= {\it val}_{M,s',\beta}(\varphi) \text{ where } s' = \val(u) \\
  \val(\dlbox{\tt p}{\varphi}) &= \begin{cases}
      {\it val}_{M,s',\beta}(\varphi) & \text{if } \val({\tt p}) = \{s'\} \\
      \semtrue                        & \text{otherwise}
    \end{cases}
\end{array}\]
\caption{Definition of the evaluation function $\val$ for formulas.}
\label{fig:evaluation-of-formulas}
\end{figure}

\begin{figure}
\[\begin{array}{rl}
  \val(\code{x=}t) &= \{s'\} \text{ where } s'({\tt y}) = \begin{cases}
      \val(t)    & \text{if } {\tt y} = {\tt x} \\
      s({\tt y}) & \text{otherwise}
    \end{cases} \\
 \val(\code{p}_1\code{;p}_2) &= \begin{cases}
     {\it val}_{M, s', \beta}({\tt p}_2) & \text{if } \val({\tt p}_1) = \{ s' \} \\
     \emptyset                           & \text{otherwise}
   \end{cases}\\
 \val(\code{if (}\varphi\code{) \{ p}_1\code{ \} else \{ p}_2\code{ \}}) &= \begin{cases}
     \val({\tt p}_1) & \text{if } \val(\varphi) = \semtrue \\
     \val({\tt p}_2) & \text{otherwise}
   \end{cases} \\
 \val(\code{while (}\varphi\code{) \{p\}}) & = \left\{
     \begin{array}{ll}
         \{ s' \}  & \text{if } \exists n\geq 0:~s=s_0, s'=s_n \text { and } \forall i < n:\\
                   & {\it val}_{M, s_i, \beta}({\tt p}) = \{s_{i+1}\} \text { and } {\it val}_{M, s_i, \beta}(\varphi) = \semtrue\\
                   & \text{and } {\it val}_{M, s_n, \beta}(\varphi) = \semfalse \\
         \emptyset & \text{otherwise }
     \end{array}\right. 
\end{array}\]
\caption{Definition of the evaluation function $\val$ for programs.}
\label{fig:evaluation-of-programs}
\end{figure}

\begin{figure}
\[\begin{array}{rl}
  \val(\elUp{{\tt x}_1}{t_1} \parUp{}{} \ldots \parUp{}{} \elUp{{\tt x}_n}{t_n}) =&
    \{ {\tt x} \mapsto s({\tt x}) \mid {\tt x} \not\in \{ {\tt x}_1, \ldots, {\tt x}_n \} \}\ \cup \\
    & \{ {\tt x} \mapsto \val(t_k) \mid {\tt x} = {\tt x}_k \text{ and } {\tt x} \not\in \{{\tt x}_{k+1}, \ldots, {\tt x}_n\} \}
\end{array}\]
\caption{Definition of the evaluation function $\val$ for updates.}
\label{fig:evaluation-of-updates}
\end{figure}

\subsection{Computation-Tracking Semantics}
\label{app:semantics:extended}

%

Figure~\ref{fig:evaluation:extended:progs} defines the updated evaluation of programs; where we focus on the changes to the dependency variables.

For assignments we wrap the term using the syntactic function $\wrap$ shown in Figure~\ref{fig:wrap}.
This function homormophically replaces all symbols by their Herbrand counterparts and all program variables by their corresponding dependency variable.

For conditional statements we evaluate both branches, select a term from the equivalence class resulting from each branch and combine them in a single conditional term.
Here, $\mathcal{C}$ is the choice operator that selects any term from an equivalence class, i.e. $\mathcal{C}(e) \in e$ for all $e : \ECS$.

For while loops we introduce an infinite set of function symbols $\mathcal{W}$, such that $\mathcal{W}({\tt y}, \bar{\tt z})$ returns the value of variable {\tt y} at the end of the while loop, given the initial values $\bar{\tt z}$ of all variables occuring in the condition $\varphi$ or body {\tt p} of the while loop.
The $\mathcal{W}$ functions are introduced only to have a well-defined semantics that tracks the dependencies. The function symbols are not used in the calculus.

\begin{figure}
\[\begin{array}{rl}
  \val(\EC{t}) &= \ECsem{t} \\
  \val(\code{x=} t) &= \{s'\} \text{ where } s'({\tt y}) = \begin{cases}
      \val(t)     & \text{if } {\tt y} = {\tt x} \\
      \val(\wrap(t))      & \text{if } {\tt y} = \dep{\tt x} \\
      s({\tt y})  & \text{otherwise}
    \end{cases} \\
 \val(\code{if (}\varphi\code{) \{ p}_1\code{ \} else \{ p}_2\code{ \}}) &= \begin{cases}
     S'_1 & \text{if } \val(\varphi) = \semtrue \\
     S'_2 & \text{otherwise}
   \end{cases} \\
  & \quad
    \text{where } S_i = \val({\tt p}_i), S'_i = \emptyset \text{ iff } S_i = \emptyset,
  \\
  & \quad
    \text{otherwise } S_i = \{ s_i \} \text{ and } S'_i = \{ s'_i \} \text{ with }
  \\
  & \quad
    s'_i({\tt y}) = \begin{cases}
        \ECsem{{\it if}(\varphi'){\it then}(t_1){\it else}(t_2)} & \text{if } {\tt y} = \dep{{\tt x}}, \\
        & \varphi' = \wrap(\varphi), \\
        & t_i = \choice(s_i(\dep{{\tt x}}))\\
        s_i({\tt y}) & \text{otherwise}
      \end{cases}
  \\
 \val(\code{while (}\varphi\code{) \{p\}}) & = \left\{
     \begin{array}{ll}
         \{ s'' \}  & \text{if } \exists n\geq 0:~s=s_0, s'=s_n \text { and } \forall i < n:\\
                   & {\it val}_{M, s_i, \beta}({\tt p}) = \{s_{i+1}\} \text { and } {\it val}_{M, s_i, \beta}(\varphi) = \semtrue\\
                   & \text{and } {\it val}_{M, s_n, \beta}(\varphi) = \semfalse \text{ with } \\
                   & s''({\tt x}) = \left\{
                     \begin{array}{ll}
                       \EC{\mathcal{W}({\tt y}, \bar{\tt z})} & \text{if } {\tt x} = \dep{\tt y} \\
                       s'({\tt x}) & \text{otherwise}
                     \end{array}\right. \\
         \emptyset & \text{otherwise }
     \end{array}\right. 
\end{array}\]
\caption{Updated definition of the evaluation function $\val$ for dependency tracking.}
\label{fig:evaluation:extended:progs}
\end{figure}

 \begin{figure}
 \begin{align*}
     \wrap({\java{x}}) ={}& \dep{\java{x}} 
  \\ \wrap(f(t_1,\dots,t_n) ={}& \EC{f} ( \wrap(t_1), \dots \wrap(t_n) )
  \\ \wrap(\condterm{\varphi}{t_1}{t_2}) ={}& \EC{\condterm{}{}{}} (\wrap(\varphi), \wrap(t_1), \wrap(t_2))
  \\ \wrap(\keytrue) ={}& \EC{\keytrue}
  \\ \wrap(\keyfalse) ={}& \EC{\keyfalse}
  \\ \wrap(p(t_1,\dots,t_n) ={}& \EC{p} ( \wrap(t_1), \dots \wrap(t_n) )
  \\ \wrap(\keynot \varphi) ={}& \EC{\keynot} ( \wrap(\varphi) )
  \\ \wrap(\varphi_1 \circ \varphi_2) ={}& \EC{\circ} (\wrap(\varphi_1), \wrap(\varphi_2))
  \end{align*}
  \caption{Definition of the syntactic operator $\wrap$.}
  \label{fig:wrap}
\end{figure}

%% file: app/proofs.tex
\section{Proofs}
\label{app:proofs}

\subsection{Equivalence between non-interference and dependency-based non-interference}
\label{app:proofs:lem:niasdep}

Non-interference implies dependency-based non-interference:

By contraposition.
Assume that there exists a program variable {\tt x} which has a variable ${\tt y} \in \depset({\tt x}, {\tt p})$ with $\lvl({\tt y}) \not\sqsubseteq \lvl({\tt x})$.
Therefore, program {\tt p} does not satisfy dependency-based non-interference (Definition~\ref{def:depnoninf}).
Since ${\tt y} \in \depset({\tt x}, {\tt p})$, by Definition~\ref{def:termdeps} there must be two states $s_1$, $s_2$ such that $s_1 \approx_{\lvl({\tt x})} s_2$, $s_1({\tt y}) \not= s_2({\tt y})$, $\mathit{val}_{M,s_i,\beta}({\tt p}) = \{s_i'\}$, and $s_1'({\tt x}) \not= s_2'({\tt x})$.
Consequently, program {\tt p} also does not satisfy non-interference (Definition~\ref{def:noninf}).

\noindent
Dependency-based non-interference implies non-interference:

Given a program {\tt p} and security level $l$.
If a program {\tt p} is dependency-based non-interferent then, by Definition~\ref{def:depnoninf} for each variable {\tt x} with level $l' \sqsubseteq l$, for all ${\tt y} \in \depset({\tt x}, {\tt p})$ we have $\lvl({\tt y}) \sqsubseteq \lvl({\tt x})$.
Given two states $s_1$, $s_2$ such that $s_1 \approx_l s_2$, it thus follows that $s_1({\tt y}) = s_2({\tt y})$ for all ${\tt y} \in \depset({\tt x}, {\tt p})$.
If $\mathit{val}_{M,s_i,\beta}({\tt p}) = \{s_i'\}$ then by Definition~\ref{def:termdeps} $s_1'({\tt x}) = s_2'({\tt x})$ for each variable {\tt x} with level $l' \sqsubseteq l$, and hence $s_1' \approx_l s_2'$.
Therefore, program {\tt p} is non-interfering (Definition~\ref{def:noninf}).

\subsection{Equivalence with declassifications}
\label{app:proofs:lem:declniasdep}

Non-interference implies dependency-based non-interference:

By contraposition.
Assume that there exists a program variable {\tt x} which has a variable ${\tt y} \in \depset({\tt x}, {\tt p}, \declpairs)$ with $\lvl({\tt y}) \not\sqsubseteq \lvl({\tt x})$.
Therefore, program {\tt p} does not satisfy dependency-based non-interference (Definition~\ref{def:depnoninf:decl}).
Since ${\tt y} \in \depset({\tt x}, {\tt p}, \declpairs)$, by Definition~\ref{def:termdeps:decl} there must be two states $s_1$, $s_2$ such that $s_1 \approx_{\lvl({\tt x})} s_2$, $s_1 \approx_{\declpairs} s_2$, $s_1({\tt y}) \not= s_2({\tt y})$, $\mathit{val}_{M,s_i,\beta}({\tt p}) = \{s_i'\}$, and $s_1'({\tt x}) \not= s_2'({\tt x})$.
Consequently, program {\tt p} also does not satisfy non-interference (Definition~\ref{def:noninf:decl}).

\noindent
Dependency-based non-interference implies non-interference:

Given a program {\tt p} and security level $l$.
If a program {\tt p} is dependency-based non-interferent then, by Definition~\ref{def:depnoninf:decl} for each variable {\tt x} with level $l' \sqsubseteq l$, for all ${\tt y} \in \depset({\tt x}, {\tt p}, \declpairs)$ we have $\lvl({\tt y}) \sqsubseteq \lvl({\tt x})$.
Given two states $s_1$, $s_2$ such that $s_1 \approx_l s_2$ and $s_1 \approx_\declpairs s_2$, it thus follows that $s_1({\tt y}) = s_2({\tt y})$ for all ${\tt y} \in \depset({\tt x}, {\tt p}, \declpairs)$.
If $\mathit{val}_{M,s_i,\beta}({\tt p}) = \{s_i'\}$ then by Definition~\ref{def:termdeps:decl} $s_1'({\tt x}) = s_2'({\tt x})$ for each variable {\tt x} with level $l' \sqsubseteq l$, and hence $s_1' \approx_l s_2'$.
Therefore, program {\tt p} is non-interfering (Definition~\ref{def:noninf:decl}).

\subsection{Correctness of $\dep{\tt x}$}
\label{app:proofs:lem:xdepcorrect}

Auxiliary lemmas:

\begin{lemma}
  For all programs {\tt p} and program variables {\tt x}, 
  for all state $s_1, s_2$, first-order structure $M$, and variable assignment $\beta$
  with $\mathit{val}_{M,s_1,\beta}({\tt p}) = \{s_1'\}$ and $\mathit{val}_{M,s_2,\beta}({\tt p}) = \{s_2'\}$,
  we have $s_1(\dep{{\tt y}}) = s_2(\dep{{\tt y}})$ for all {\tt y} $\in \PVars$ implies $s_1'(\dep{{\tt x}}) = s_2'(\dep{{\tt x}})$.
  \label{app:lemma:depsame}
\end{lemma}
\begin{proof}
  By inspection of $\val$ we observe that the state $s$ is irrelevant for determining the value of $\dep{\tt x}$, except for the evaluation of conditional statement and program variables.
  In these cases, the value of other dependency variables is used to construct a new value.
  Since we have that these values are the same in any two states (for the conditional statement by induction) also the newly constructed value is equal in any two states.
\end{proof}

\begin{lemma}
  For all programs {\tt p} and program variables {\tt x}, 
  for any state $s$, first-order structure $M$, and variable assignment $\beta$
  with $\val({\tt p}) = \{s'\}$ and $s(\dep{\tt y}) = \ECsem{\tt y}$ for all ${\tt y} \in \PVars$,
  it holds that $\val(t) = s'({\tt x})$ for each $t \in s'(\dep{\tt x})$.
  \label{app:lemma:depcorrect}
\end{lemma}
\begin{proof}
  By induction on the structure of {\tt p}.
  \begin{itemize}
  
    \item Case $\java{x = } t$:
    
          Here $s'({\tt x}) = \val(t)$ and $s'(\dep{{\tt x}}) = \wrap(t)$ with each program variables {\tt y} substituted with $\ECsem{\tt y}$.
          By definition of $\eqterm$, all terms in this equivalence class evaluate to the same value in state $s$, so to $\val(t)$.
          
    \item Case ${\tt p}_1{\tt ; p}_2$:
    
          We have $\val({\tt p}_1) = \{ s' \}$ and $\mathit{val}_{M,s',\beta}({\tt p}_2) = \{ s'' \}$.
          Let $r'$ be like $s'$ but with $r'(\dep{{\tt y}}) = \ECsem{\tt y}$ for all ${\tt y} \in \PVars$, such that $\mathit{val}_{M,r',\beta}({\tt p}_2) = \{ r'' \}$.
          
          By induction, $\mathit{val}_{M,r',\beta}(t_{\tt x}) = r''({\tt x})$ for all $t_{\tt x} \in r''(\dep{{\tt x}})$.
          We also have that since $r'$ differs from $s'$ only in the dependency variables, $r''({\tt x}) = s''({\tt x})$.
          Therefore, $\mathit{val}_{M,r',\beta}(t_{\tt x}) = s''({\tt x})$ for all $t_{\tt x} \in r''(\dep{{\tt x}})$.
          
          We also have by induction that $\mathit{val}_{M,s,\beta}(t_{\tt y}) = s'({\tt y})$ for all $t_{\tt y} \in s'(\dep{{\tt y}})$.
          We can therefore substitute each variable {\tt y} in $t_{\tt x}$ with any term $t_{\tt y} \in s'(\dep{{\tt y}})$ and evaluate $t_{\tt x}$ in state $s$ instead.
          That is, $\mathit{val}_{M,s,\beta}(t_{\tt x}[{\tt y} \mapsto t_{\tt y}]) = s''({\tt x})$ for all $t_{\tt x} \in r''(\dep{{\tt x}})$ and all $t_{\tt y} \in s'(\dep{{\tt y}})$.
          
          We finally observe that the $s''(\dep{{\tt x}})$ equivalence class is exactly the set $t_{\tt x}[{\tt y} \mapsto t_{\tt y}]$ for all $t_{\tt x} \in r''(\dep{{\tt x}})$ and all $t_{\tt y} \in s'(\dep{{\tt y}})$.
          Therefore,  $\val(t) = s''({\tt x})$ for each $t \in s''(\dep{\tt x})$.
    
    \item Case $\code{if (}\varphi\code{) \{ p}_1\code{ \} else \{ p}_2\code{ \}}$:
    
          Let $\val({\tt p}_1) = \{ s_1' \}$ and $\val({\tt p}_2) = \{ s_2' \}$.
          We have that $s'({\tt x}) = s_1'({\tt x})$ if $\val(\varphi) = \semtrue$ and $s'({\tt x}) = s_2'({\tt x})$ otherwise.
          
          For all terms $t \in s'(\dep{{\tt x} })$ we have $t \eqterm \condterm{\varphi}{t_1}{t_2}$, with (by induction) $\val(t_1) = s_1'({\tt x})$ and $\val(t_2) = s_2'({\tt x})$.          
          If $\val(\varphi) = \semtrue$ we have that $\val(\condterm{\varphi}{t_1}{t_2}) = \val(t_1) = s_1'({\tt x}) = s'({\tt x})$, otherwise we have that $\val(\condterm{\varphi}{t_1}{t_2}) = \val(t_2) = s_2'({\tt x}) = s'({\tt x})$.
          
    \item Case $\code{while (}\varphi\code{) \{p\}}$:
    
          Directly by definition of the $\mathcal{W}$ function symbols, for all terms in $t \in s'(\dep{{\tt x}})$ we have $\val(t) = s'({\tt x})$.
  \end{itemize}
\end{proof}

\begin{lemma}
  For all well-typed terms $t$; $\depset(t, {\tt ;}) \subseteq \vars(t)$.
  \label{app:lemma:varsubset}
\end{lemma}
\begin{proof}
  Our valuation function $\val$ determines the value of terms only by the program variables that actually occur in that term.
  Therefore, all dependencies of a term are present in the syntactic representation of that term.
\end{proof}

\noindent
We can now show that for all programs {\tt p} and program variables {\tt x}, 
for any state $s$, first-order structure $M$, and variable assignment $\beta$
with $\val({\tt p}) = \{s'\}$ and $s(\dep{\tt y}) = \ECsem{\tt y}$ for all ${\tt y} \in \PVars$,
it holds that $\depset({\tt x}, {\tt p}) \subseteq \vars(t)$ for each $t \in s'(\dep{\tt x})$.

\begin{proof}
  By Lemma~\ref{app:lemma:varsubset}, this is implied by showing that $\depset({\tt x}, {\tt p}) \subseteq \depset(t, {\tt ;})$.
  By definition of $\depset$; for any two states $s_1$, $s_2$ we have if $s_1({\tt y}) = s_2({\tt y})$ for all ${\tt y} \in \depset(t, {\tt ;})$, then $\mathit{val}_{M,s_1,\beta}(t) = \mathit{val}_{M,s_2,\beta}(t)$.
  
  Let $\mathit{val}_{M,s_i,\beta}({\tt p}) = \{ s_i' \}$.
  By Lemma~\ref{app:lemma:depsame} we have that $s'(\dep{{\tt x}}) = s_1'(\dep{{\tt x}}) = s_2'(\dep{{\tt x}})$, therefore $t \in s_i'(\dep{{\tt x}})$.  
  By Lemma~\ref{app:lemma:depcorrect} we conclude that $\mathit{val}_{M,s_i,\beta}(t) = s_i'({\tt x})$, and since we had $\mathit{val}_{M,s_1,\beta}(t) = \mathit{val}_{M,s_2,\beta}(t)$ we can conclude that $s_1'({\tt x}) = s_2'({\tt x})$.
  Therefore $\depset(t,{\tt ;})$ is a correct over-approximation of the dependencies of {\tt x} under {\tt p}, that is $\depset({\tt x}, {\tt p}) \subseteq \depset(t, {\tt ;})$ (by definition of $\depset$) which is what we had to show.
\end{proof}

\subsection{Soundness}
\label{app:proofs:soundness}

Given:
\[
  {\it INIT} \keyand
  {\it DEPS} \keyimplies
  [{\tt p}] \deps(\dep{\tt x})) \subseteq {\it LVL}({\tt x})
\]
If the declassification environment is empty, then {\it DEPS} is $\forall \ECS e.\ \deps(e) \keyeq h(e)$.
Then $\deps(\dep{\tt x}) = \vars(\mathcal{C'}(\dep{\tt x})))$ where $\mathcal{C'}$ selects the syntactic term that happens to represent \dep{\tt x},
which by  Lemma~\ref{lem:xdepcorrect} gives us that $\depset({\tt x}, {\tt p}, \declenv) \subseteq {\it LVL}({\tt x})$ which implies non-interference with declassification.

If the declassification environment is not empty, for each declassification pair $(\varphi, t)$ then some program variables occurring in $t$ may be removed from the set returned by $\deps$.
However, this only happens if $\varphi$ holds, therefore these would also be excluded by definition of $\depset$ and still imply non-interference with declassification.